\documentclass[letterpaper, 10 pt, conference]{ieeeconf}  
\usepackage{amsmath,amssymb,graphicx,epsfig,cite,fancyhdr,amsthm}
\usepackage{subfigure}
\usepackage{url,hyperref}
\usepackage{tikz}
\usepackage{multicol}
\usetikzlibrary{calc}
\usepackage{pdfsync}
\usepackage{flushend} 

\makeatletter
\newcommand*\dashline{\rotatebox[origin=c]{90}{$\dabar@\dabar@\dabar@\dabar@\dabar@\dabar@\dabar@$}}
\makeatother

\newtheorem{thm}{Theorem}[]
\newtheorem{cor}{Corollary}
\newtheorem{lem}{Lemma}

\theoremstyle{remark}

\theoremstyle{definition}
\newtheorem{defn}{Definition}

\begin{document}

\title{\huge  The Capacity of Less Noisy Cognitive Interference Channels}

\author{\normalsize Mojtaba Vaezi
\\Department of Electrical and Computer Engineering\\
McGill University\\
Montreal, Quebec H3A 0C6, Canada\\
Email: mojtaba.vaezi@mail.mcgill.ca
}

%


\maketitle


\begin{abstract}

Fundamental limits of the {\it cognitive interference channel} (CIC) with
two pairs of transmitter-receiver have been under exploration for several years.
In this paper, we study the discrete memoryless cognitive interference channel (DM-CIC) in which the cognitive
transmitter non-causally knows the full message of the
primary transmitter. The capacity of this channel is not known in general; it is only known in some special cases.
Inspired by the concept of less noisy
broadcast channel (BC), in this work we introduce the notion of {\it less noisy}
cognitive interference channel.
Unlike BC, due to the inherent
asymmetry of the cognitive channel, two different less noisy channels are distinguishable; these are named the {\it primary-less-noisy}
and {\it cognitive-less-noisy} channels. We derive capacity region for the latter case
by introducing inner and outer bounds on the capacity
of the DM-CIC and showing that these bounds coincide for the cognitive-less-noisy channel.
%
Having established the capacity
region, we prove that superposition coding is the optimal encoding technique. 
\end{abstract}

\IEEEpeerreviewmaketitle

\section{Introduction}\label{sec:intro}
A two-user interference channel (IC) is a network consisting of two
transmitter-receiver pairs, communicating
over the same channel, and thus interfering each other. In certain
communication scenarios, e.g., cognitive radio, one transmitter
(the cognitive transmitter) is able to sense the environment and
obtain side information about the incumbent transmitter (the primary transmitter).
 Such a communication channel is called interference channel
with cognition or simply the {\it cognitive channel}.
Motivated by cognitive radio's promise for increasing the spectral
efficiency in wireless systems, the study of interference channel
with cognitive users has been receiving increasing attention during the past years.

Fundamental limits of the cognitive
interference channel, in which the cognitive transmitter non-causally
knows the the full message of the the primary user, has been studied in
\cite{Devroye,Jovicic-Viswanath,Wu-Vishwanath,Maric1,Maric2,Maric3,Rini2,vaezi2011capacity,JiangZ,vaezi2011superposition,Rini3,liu2009bounds}.
This channel was first introduced in \cite{Devroye} where the authors obtained achievable rates by applying
Gel'fand-Pinsker coding \cite{Gel�fand-Pinsker} to the
celebrated Han-Kobayashi encoding \cite{Han-Kobayashi} for the
IC. The capacity of this channel remains unknown in general; however, it is known
in several special cases, both in the discrete memoryless and Gaussian channels.

Capacity of the Gaussian cognitive interference channel (GCIC) is
known at low interference \cite{Jovicic-Viswanath} and \cite{Wu-Vishwanath},
as well as strong interference  \cite{Maric1}.
Besides, capacity of Gaussian cognitive Z-interference channel (GCZIC) in
which the primary receiver is interfered by the cognitive
transmitter is known for several ranges of interference gain
\cite{vaezi2011capacity,JiangZ,vaezi2011superposition,Rini3}. While at low interference
dirty paper coding \cite{Costa} is capacity-achieving scheme, at high
interference superposition coding is the optimal technique.
For the discrete memoryless channel, capacity
is known for ``strong interference'' \cite{Maric1}, ``weak interference'' \cite{Wu-Vishwanath},
and ``better cognitive decoding'' \cite{Rini2} regimes.
Effectively, superposition coding is the capacity-achieving technique in all above cases
although several other techniques, including rate-splitting, simultaneous coding,
and Gel'fand-Pinsker coding (binning) are used to find achievable
rate regions.

In this paper, we consider the {\it discrete memoryless} cognitive
interference channel (DM-CIC). We first introduce the notion
of {\it less noisy} DM-CIC and show that there are two different less noisy
cognitive channels: the {\it primary-less-noisy} and {\it cognitive-less-noisy} DM-CIC.
In the former, the primary receiver is
less noisy than the secondary receiver, whereas it is the opposite in the latter.

Afterward, we propose two inner bounds for the DM-CIC; one based on superposition coding,
and another one using independent coding.
Obviously, these inner bounds are also valid
for less noisy DM-CIC; in fact, one of these inner bounds is more suitable for the
primary-less-noisy DM-CIC whereas the other one is better for the cognitive-less-noisy  DM-CIC.
We also prove an outer bound on the capacity of this channel.

Finally, we show that for the cognitive-less-noisy DM-CIC
the inner and outer bounds coincide, and therefore we establish the capacity region for this class of
DM-CIC. This proves that superposition coding is the capacity-achieving scheme in the
less noisy DM-CIC, as it is in the less noisy BC.
Although for the primary-less-noisy DM-CIC capacity remains unknown, corresponding
inner bound simplifies to an achievable region that has already been proved to be capacity-achieving in the special case of GCZIC
\cite{vaezi2011capacity}, \cite{vaezi2011superposition}.

This paper is organized as follows. In Section \ref{sec:models},
we introduce the system model and define the less noisy DM-CIC.
In Section \ref{sec:DM-CIC}, we propose an outer bound and two inner bounds
for the DM-CIC. Then, in Section \ref{sec:cap}, we show
that one of the inner bounds is tight for the cognitive-less-noisy channel,
and thus provides capacity for this class of the DM-CIC.
New capacity result is compared with the existing ones in Section~\ref{sec:dis}.


\section{Problem Setup and Definitions}
\label{sec:models}

The two-user discrete memoryless cognitive interference channel
(DM-CIC) is an interference channel \cite{Carleial} that consists of two
transmitter-receiver pairs, in which one transmitter (the cognitive user)
knows the message of the other transmitter (the primary one), in addition to its own message.
In what follows, we formally
define this channel and a special class of that.

\subsection{Discrete Memoryless Cognitive Interference Channel}

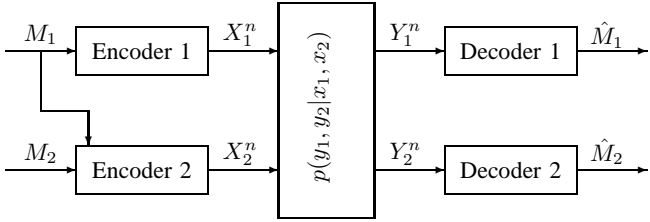
\begin{figure}
\begin{center}
\scalebox {.9}{
\begin{picture}(280,170)
\put(0,50){
\begin{picture}(200,80)
\put(20,20){\framebox(55,20){Encoder 2}}
\put(-10,30){\vector(1,0){30}}
\put(12,37){\makebox(0,0)[r]{$M_{2}$}}
\put(75,30){\vector(1,0){30}}
\put(96,37){\makebox(0,0)[r]{$X^n_{2}$}}
\put(105,10){\framebox(40,90)}
\put (117,25){\rotatebox{90}{$p(y_1,y_2|x_1,x_2)$}}
\put(145,30){\vector(1,0){30}}
\put(165,37){\makebox(0,0)[r]{$Y^n_{2}$}}
\put(175,20){\framebox(55,20){Decoder 2}}
\put(230,30){\vector(1,0){30}}
\put(250,38){\makebox(0,0)[r]{$\hat{M}_{2}$}}
\end{picture}
}

\put(0,100){
\begin{picture}(200,80)
\put(20,20){\framebox(55,20){Encoder 1}}
\put(-10,30){\vector(1,0){30}}
\put(12,37){\makebox(0,0)[r]{$M_{1}$}}
\put(75,30){\vector(1,0){30}}
\put(96,37){\makebox(0,0)[r]{$X^n_{1}$}}
\put(145,30){\vector(1,0){30}}
\put(165,37){\makebox(0,0)[r]{$Y^n_{1}$}}
\put(175,20){\framebox(55,20){Decoder 1}}
\put(230,30){\vector(1,0){30}}
\put(250,37){\makebox(0,0)[r]{$\hat{M}_{1}$}}
\put(5,30){\line(0,-1){25}} \put(5,5){\line(1,0){20}}
\put(25,5){\vector(0,-1){15}}
\end{picture}
}
\end{picture}
}
\end{center}

\vspace{-35pt}

\caption{The discrete memoryless cognitive
interference channel (DM-CIC) with two transmitters and two receivers.  $M_1, M_2$ are two messages,
$X_1, X_2$ are inputs, $Y_1, Y_2$ are outputs, and $p(y_1,y_2|x_1,x_2)$ is the transition probability of channel.}
  \label{fig:DM-CZIC}
\end{figure}

The discrete memoryless cognitive interference channel
(DM-CIC) is depicted in Figure~\ref{fig:DM-CZIC}. Let $M_1$ and $M_2$
be two independent messages which are uniformly distributed on the set of
all messages for the first and second users, respectively.  Transmitter $i, i\in \{1,2\},$ wishes to transmit message $M_i$
to receiver $i$, in $n$ channel use at rate $R_i$. Message $M_2$ is available only at transmitter 2, while both
transmitters know $M_1$. This channel is defined by a tuple $({\cal
X}_1,{\cal X}_2;p(y_1,y_2|x_1,x_2);{\cal Y}_1,{\cal Y}_2)$ where
${\cal X}_1,{\cal X}_2$ and ${\cal Y}_1,{\cal Y}_2$
are input and output alphabets, and $p(y_1,y_2|x_1,x_2)$ is channel transition probability density
functions.


The capacity of the DM-CIC is known in ``strong interference'' \cite{Maric1},
``weak interference'' \cite{Wu-Vishwanath}, and ``better cognitive decoding'' \cite{Rini2} regimes.
These capacity results are listed in Table~\ref{table1}, and labeled $\mathcal C_1$, $\mathcal C_2$,
and $\mathcal C_3$, respectively.
In the first case, both receivers can decode both messages. In all above cases,
the cognitive receiver has a better condition (more information) than the primary one, in some sense,
as it is evident from corresponding conditions in Table~\ref{table1}.

\subsection{Less Noisy DM-CIC}

Since the second transmitter has complete and non-causal
knowledge of both messages, it can act like a
BC transmitter.
Particularly, in the absence of the first transmitter this channel becomes the
well-known DM-BC \cite{Cover}. In the presence of that, this
channel is no longer a BC; however, one can define
conditions, similar to that in the BC, showing that one receiver is in a ``better''
condition than the other to decode the messages, e.g., one receiver is {\it less noisy}
or {\it more capable} than the other \cite{ElGamal}, \cite{ElGamalMoreCapable}.

In \cite{vaezi2011capacity}, \cite{vaezi2011superposition}, the authors extended this
notion to the DM-CIC, and studied
the case where the primary receiver is more capable than the secondary receiver.
This led to the capacity of the GCZIC at very strong interference.
In what follows, we introduce the notion of less noisy cognitive interference channel, and show that
two different less noisy DM-CIC arises, depending on which receiver is in a better condition.
These are formally defined in the following.

\begin{defn}
The DM-CIC is said to be {\it primary-less-noisy} if
\begin{align}
I(U;Y_1) \geq I(U;Y_2)
\label{eq:defn1}
\end{align}
\label{cond1}
 for all $p(u,x_1,x_2)$.
\end{defn}

\begin{defn}
The DM-CIC is said to be {\it cognitive-less-noisy} if
\begin{align}
I(U;Y_2) \geq I(U;Y_1)
\label{eq:defn2}
\end{align}
\label{cond2}
 for all $p(u,x_1,x_2)$.
\end{defn}

\noindent It is clear that in the first case the primary
receiver is less noisy than the cognitive receiver whereas in the second case
the cognitive receiver is less noisy than the primary receiver.
Therefore, given the channel condition, a DM-CIC can be {\it primary-less-noisy}, {\it cognitive-less-noisy},
 neither of them or both.


\section{Inner and Outer Bounds for the DM-CIC }
\label{sec:DM-CIC}

In this section, we first introduce an
outer bound on the capacity of the DM-CIC; we then derive two achievable rate regions
for this channel.
The first achievable region is based on superposition coding technique;
it is inspired by the capacity-achieving superposition coding in the less noisy and more capable DM-BC,
or the inner bound introduced for the more capable DM-CIC in \cite{vaezi2011superposition}.
The idea of outer bound also comes from the capacity of the
less noisy DM-BC. However, we combine two different bounds to find a unified one.

\subsection{A Unified Outer Bound}
Inspired by capacity of less noisy BC \cite{ElGamal}, and definitions \eqref{eq:defn1} and
\eqref{eq:defn2} for less noisy cognitive interference channels,
 we present a simple outer bound on the capacity of the DM-CIC. This outer bound
 is in fact a combination of two simpler outer bounds as we describe later in this section.
 Each outer bound can be tight in specific cases of less noisy DM-CIC, as it will be shown later.

The following provides an outer bound on the capacity of the DM-CIC.
\begin{thm}
The union of rate pairs $(R_{1},R_{2})$ such that
\begin{align}
  R_1 &\leq I(U;Y_1),  \label{eq:O-1}  \\
  R_2 &\leq I(V;Y_2),  \label{eq:O-2} \\
  R_1 + R_2 &\leq I(X_2;Y_2|U) + I(U;Y_1), \label{eq:O-3} \\
  R_1 + R_2 &\leq I(X_1;Y_1|V) + I(V;Y_2), \label{eq:O-4}
\end{align}
for some joint distribution $p(u,v,x_1,x_2)$ gives an
outer bound on the capacity region of the DM-CIC.
\label{thm1}
\end{thm}

\begin{proof}
The proof of the second and last inequalities follows the same line of argument
 as in the outer bound of the more capable DM-CIC \cite[Theorem 2]{vaezi2011superposition}, or similarly
 the converse of the more capable BC \cite{ElGamalMoreCapable}. The other two inequalities, by symmetry, follow the
same line of proof.
The essence of the proof in \eqref{eq:O-3} and \eqref{eq:O-4} is to use the Csiszar sum identity
and the auxiliary random variables $U_i = (M_1, Y^{i-1}_2,Y^n_{1,i+1})$ and
$V_i = (M_2, Y^{i-1}_{1}, Y^n_{2,i+1})$. The choice of $U_i, V_i$ indicates that they are
 correlated; hence, the outer bound is over the joint distribution $p(u,v)p(x_1,x_2|u,v)p(y_1,y_2|x_1,x_2)$.
\end{proof}

The symmetry of the outer bound indicates how it
consists of two simpler outer bounds. One including \eqref{eq:O-1} and \eqref{eq:O-3}, and
the other including \eqref{eq:O-2} and \eqref{eq:O-4}.
Each outer bound is resembling the capacity of less noisy DM-BC \cite{ElGamal}.

\subsection{New Achievable Rate Regions}
\label{inner}
We next provide two achievable rate regions for the DM-CIC.
The first achievable region uses superposition encoding at the
cognitive transmitter whereas the second one encodes independently.
The decoding is based on the joint typicality in both cases.

\begin{thm}
The union of rate regions given by
\begin{equation}
\begin{aligned}
    R_1 &\leq I(W,X_1;Y_1),   \\
R_2 &\leq I(X_2;Y_2|W,X_1),    \\
  R_1 + R_2 &\leq I(X_1,X_2;Y_2),
  \label{eq:inner2}
\end{aligned}
\end{equation}
is achievable for the DM-CIC, where the union is over all probability distributions $p(w,x_1, x_2)$.
\label{thm2}
\end{thm}

\begin{proof}
The proof of Theorem~\ref{thm2} uses the superposition coding idea in which
$Y_1$ can only decode $M_1$
while $Y_2$ is intended to decode both $M_1$ and $M_2$.
Considering the space of all codewords, one can
view the $(W, X_1)$ as {\it cloud centers}, and the $X_2$ as {\it satellites} \cite{Kramer}.
For completeness, the details of the proof are provided in Section~\ref{anx1}.
\end{proof}
In light of the above discussion, we expect the encoding scheme in Theorem~\ref{thm2}
be more favorable when the second receiver is in a better situation than the first one, because it
can decode both cloud centers and satellites. If the channel condition
is the reverse, i.e., the first receiver has a better situation than the second receiver,
it makes sense to reverse the order of encoding. However, at the first transmitter, we cannot do superposition encoding
against the codeword of the secondary transmitter because the first transmitter does not know the massage of
the cognitive user. As a result, the input distribution needs to be independent
as proposed in the following theorem.

\begin{thm}
The union of rate regions given by
\begin{equation}
\begin{aligned}
  R_1 &\leq I(X_1;Y_1|W,X_2),   \\
  R_2 &\leq I(W,X_2;Y_2),  \\
  R_1 + R_2 &\leq I(X_1,X_2;Y_1),
  \label{eq:inner1}
\end{aligned}
\end{equation}
is achievable for the DM-CIC, where the union is over all probability distributions $p(w,x_1, x_2)$
that factors as $p(w,x_2)p(x_1)$.
\label{thm3}
\end{thm}

\begin{proof}
The proof of Theorem~\ref{thm3} uses independent encoding of $X_1$ and $(W, X_2)$; however,
$Y_1$ is intended to decode both messages whereas $Y_2$ can only decode $M_2$.
The proof of Theorem~\ref{thm3} follows a similar footsteps as Theorem~\ref{thm2}, but the input distributions are different.
The details of the proof can be found in Section~\ref{anx2}.
\end{proof}


%
%

\section{The Capacity of Less Noisy DM-CIC}
\label{sec:cap}
In this section, we simplify the inner bounds in Theorem~\ref{thm2} and Theorem~\ref{thm3}, respectively for
 the cognitive-less-noisy and primary-less-noisy DM-CIC defined in \eqref{eq:defn1} and \eqref{eq:defn2}.
 Then, by comparing the fist inner bound with the outer bound in Theorem~\ref{thm1}, we establish capacity region for the cognitive-less-noisy DM-CIC.
 \subsection{The Cognitive-less-noisy DM-CIC}
\label{sec:cap1}
\begin{thm}
For the cognitive-less-noisy DM-CIC, the capacity region is given by the set of all rate pairs $(R_1, R_2)$ such that
\begin{align}
  R_1 &\leq I(U;Y_1),\\
  R_2 &\leq I(X_2;Y_2|U),
  \label{eq:cap1}
\end{align}
for some $p(u,x_2)$.
\label{thm4}
\end{thm}

\begin{proof}
Consider the achievable region in Theorem~\ref{thm2} and define $U=(W,X_1)$.
From \eqref{eq:defn2} we know that, for the cognitive-less-noisy DM-CIC, $I(U;Y_1) \leq I(U;Y_2)$. Then, it can be simply verified that, the third inequality in
Theorem~\ref{thm2} becomes redundant for this channel. Thus, the achievability of the rate region in Theorem~\ref{thm4} immediately follows.
To prove the converse, we consider inequalities \eqref{eq:O-1} and \eqref{eq:O-3} from the outer bound in Theorem~\ref{thm1}, which are
\begin{equation}
\begin{aligned}
  R_1 &\leq I(U;Y_1),    \\
  R_1 + R_2 &\leq I(X_2;Y_2|U) + I(U;Y_1).
   \label{eq:region1}
\end{aligned}
\end{equation}
Clearly, these two inequalities make an outer bound on the capacity of
any DM-CIC for some joint distributions $p(u,x_1,x_2)p(y_1,y_2|x_1,x_2)$.
An alternative representation
of this outer bound is given by \cite{ElGamal}, \cite{ElGamalMoreCapable},
\begin{equation}
\begin{aligned}
  R_1 &\leq I(U;Y_1),    \\
  R_2 &\leq I(X_2;Y_2|U),
    \label{eq:region2}
\end{aligned}
\end{equation}
which is equal to the achievable region given in Theorem~\ref{thm4}. 
Hence, the rate region in Theorem~\ref{thm4}
is the capacity of the cognitive-less-noisy DM-CIC.
Note that the regions characterized by  \eqref{eq:region1} and \eqref{eq:region2} 
are not necessarily equal for fixed $U, X_1$; however, 
their convex hull over all $p(u,x_1)$ becomes the same.

\end{proof}
We further observe that the auxiliary random variable $U$ in the capacity region, can be replaced by $(W, X_1)$,
which results in the following corollary.

\begin{cor}
The capacity region of the cognitive-less-noisy DM-CIC can be expressed as
\begin{equation}
\begin{aligned}
    R_1 &\leq I(W,X_1;Y_1),   \\
R_2 &\leq I(X_2;Y_2|W,X_1),
     \label{eq:cap2}
\end{aligned}
\end{equation}
for some $p(w,x_1,x_2)$.
\label{cor1}
\end{cor}
\begin{proof}
The achievability of this region is obvious from Theorem~\ref{thm2}
and the condition in \eqref{eq:defn2}. To prove the converse, we use
the last two constraints of the outer bound in \cite[Theorem 3.2]{Wu-Vishwanath},
which are (note the reversal of indices),
\begin{equation}
\begin{aligned}
    R_1 &\leq I(W,X_1;Y_1),   \\
 R_1 + R_2 &\leq I(X_2;Y_2|W,X_1) +I(W,X_1;Y_1),
  \label{eq:cap3}
\end{aligned}
\end{equation}
for some $p(w,x_1,x_2)$.
However, with a similar argument used in the proof of Theorem~\ref{thm4}, the outer bound in \eqref{eq:cap3}
can be alternatively represented as the constraints in \eqref{eq:cap2}.
  \end{proof}
   The capacity-achieving technique in Theorem~\ref{thm4} is the well-known superposition coding,
similar to that in the less noisy BC \cite{ElGamal}. Superposition coding has been proved
to be optimal encoding in several other cases, both for the DM-CIC (see Table~\ref{table1}) and the GCZIC \cite{vaezi2011superposition}.

 \subsection{The Primary-less-noisy DM-CIC}
\label{sec:cap2}
One may expect a similar result for the primary-less-noisy DM-CIC, by applying
the corresponding condition in \eqref{eq:defn1} to the rate region in Theorem~\ref{thm3}.
However, since Theorem~\ref{thm3} holds only for independent $x_1$ and $x_2$,
capacity region cannot be established in general. Instead, we can write
\begin{cor}
The union of all rate pairs $(R_1 , R_2 )$ satisfying
\begin{align}
R_1 &\leq I(X_1;Y_1|V),\\
R_2 &\leq I(V;Y_2),
  \label{eq:cap1}
\end{align}
over all probability distributions  $p(v,x_1,x_2,y_1,y_2)$ that factors as $p(v)p(x_2)p(y_1,y_2|x_1,x_2)$
is achievable for the primary-less-noisy DM-CIC.
\label{cor2}
\end{cor}

\begin{proof}
By symmetry, the proof of this theorem follows the same line of argument as the proof of Theorem~\ref{thm4}. 
 To prove the achievability, define $V=(W,X_2)$ and apply the condition of
 the primary-less-noisy DM-CIC in \eqref{eq:defn1} to Theorem~\ref{thm3}; this makes the third inequality of
Theorem~\ref{thm3} redundant and completes the proof of the achievability.
\end{proof}
 Note that, from \eqref{eq:O-2} and \eqref{eq:O-4} a outer bound that resembles
 the rate region in Corollary~\ref{cor2} can be built, but this outer bound is over $p(v,x_2)$ which is, in general, larger than
 the inner bound in Corollary~\ref{cor2}. Nevertheless, in the following section
 we discuss that this region can result in capacity region for a particular channel.

\begin{table*}[tb]
\renewcommand{\arraystretch}{1.3}
\caption{Summary of the capacity results for the discrete memoryless cognitive interference channel} \label{table1}
\centering
\scalebox{.99}{
\begin{tabular}{|c|c|c|c|c|}
\hline
\bfseries Label & \bfseries Condition & \bfseries Capacity region & \bfseries Encoding  & \bfseries Reference \\
\hline\hline\

$\mathcal C_1$  &  $I(X_1,X_2;Y_1) \leq I(X_1,X_2;Y_2)$ & $
   R_1 + R_2 \leq I(X_1,X_2;Y_1) $ & superposition coding &  \cite{Maric1} \\
$ $ & $ I(X_2;Y_2|X_1) \leq I(X_2;Y_1|X_1) $ & $ R_2 \leq I(X_2;Y_2|X_1) $ &   &  \\
\hline

$\mathcal C_2$ & $ I(X_1;Y_1) \leq I(X_1;Y_2) $ & $ R_1 \leq I(U,X_1;Y_1) $ & superposition coding  & \cite{Wu-Vishwanath} \\
 $ $ &$ I(U;Y_1|X_1) \leq I(U;Y_2|X_1)  $ & $ R_2 \leq I(X_2;Y_2|U,X_1) $ &  &  \\
\hline

$$  & $ $ & $ R_1 \leq I(U,X_1;Y_1) $ & rate-splitting,\rlap{\textsuperscript{*}}  &  \\
$ \mathcal C_3$ & $  I(U,X_1;Y_1) \leq I(U,X_1;Y_2) $ & $ R_2 \leq I(X_2;Y_2|X_1) $ & binning, and & \cite{Rini2} \\
$ $ & $ $ & $ R_1 + R_2 \leq I(U,X_1;Y_1) + I(X_2;Y_2|U,X_1) $ &  superposition coding &   \\
\hline

$\mathcal C_4$  & $ I(U;Y_1) \leq I(U;Y_2) $ & $ R_1 \leq I(U;Y_1) $ & superposition coding  & Theorem~\ref{thm4} \\
$  $ & $(\mbox {cognitive-less-noisy DM-CIC})$ & $ R_2 \leq I(X_2;Y_2|U) $  &  &  \\
\hline


\end{tabular}}
\begin{flushleft}
 \qquad \qquad  \qquad \scriptsize\textsuperscript{*} It should be emphasized that the technique used to achieve $\mathcal C_3$ effectively is superposition coding, although it is
 derived (simplified) from a scheme  \\
 \qquad \qquad  \qquad  \qquad  that uses rate-splitting, binning, and superposition coding. In fact, $\mathcal C_3$ is only a different representation $\mathcal C_2$, as shown in \cite{VaeziMoreCapable}
\end{flushleft}
\end{table*}

\section{Comparison and Discussion}
\label{sec:dis}
In this section we compare the capacity region obtained in
Theorem~\ref{thm4}
with the existing capacity results for the DM-CIC.
Table I summarizes the capacity results for the DM-CIC in the chronological order.

We show that the capacity of the cognitive-less-noisy DM-CIC is
a subset of the capacity region derived in \cite{Wu-Vishwanath}, which is labeled as $\mathcal C_2$ in Table~\ref{table1}.
To this end, we first show that the condition \eqref{eq:defn2} of the cognitive-less-noisy implies both conditions
required for $\mathcal C_2$. First, since $ I(U;Y_1) \leq I(U;Y_2)$ holds for any $p(u,x_1,x_2)$, it will
result in $ I(X_1;Y_1) \leq I(X_1;Y_2)$ for $U=X_1$. The other condition is also achieved by the following lemma.

\begin{lem}
If $ I(U;Y_1) \leq I(U;Y_2)$ holds for all joint distributions $p(u,x_1,x_2)$, then
$ I(U;Y_1|X_1) \leq I(U;Y_2|X_1)$ for all $p(u,x_1,x_2)$.
\label{lem1}
\end{lem}
\begin{proof}
See Appendix~\ref{anx3}.
\end{proof}

Thus, the condition required for $\mathcal C_4$ is more demanding than that of $\mathcal C_2$. In other words,
if the cognitive receiver, in a DM-CIC, is less noisy than the primary one, the DM-CIC will satisfy the
 ``weak interference'' conditions. Further, we observe that, for $U=(U,X_1)$ the capacity regions
$\mathcal C_4$ becomes the same as $\mathcal C_2$. This is also evident from Corollary~\ref{cor1}.

It is also worth mentioning that, for $U=X_1$, with further
assumption that $ I(X_2;Y_2|X_1) \leq I(X_2;Y_1|X_1) $, $\mathcal C_4$ becomes equivalent to $\mathcal C_1$. This indicates that
we can use superposition coding to achieve the capacity of the DM-CIC in the ``strong interference'' regime.
Note that, the capacity region in the ``strong interference'' ($\mathcal C_1$ in Table~\ref{table1}),
can be reexpressed as
\begin{align}
R_1 &\leq I(X_1;Y_1),\label{cap2:second}\\
R_2 &\leq I(X_2;Y_2|X_1). \label{cap2:first}
\end{align}
In this setting, $X_1$ and $X_2$, respectively, can be viewed as cloud centers and satellites of superposition coding.
Originally, the achievability of $\mathcal C_1$ is proved by using the capacity of compound multiple accesses channels \cite{Maric2}
which is based on transmitting private and common messages.

It should be highlighted that, the technique used to achieve $\mathcal C_3$ is also effectively superposition coding although it is
 derived (simplified) from a scheme that uses rate-splitting, binning, and superposition coding collectively. This can be verified by
 looking at the simplified encoding in the proof of the achievability in \cite{Rini2}. Therefore, we can see that
 all capacity results in Table~\ref{table1} ($\mathcal C_1 - \mathcal C_4$) can be achieved using superposition
 coding.\footnote{We should emphasis that $\mathcal C_3$ is just a different representation of $\mathcal C_2$;
 this is because the conditions required for these two capacity regions are equal. This is proved in \cite{VaeziMoreCapable}.}

Finally, consider the primary-less-noisy DM-CIC. The condition required for
this channel is rather different from that in all other cases that we know the capacity region, and listed in Table~\ref{table1}. To appreciate this,
from Table~\ref{table1}, one can see that in all those cases ($\mathcal C_1 - \mathcal C_4$)
the cognitive receiver has, in some sense, more information than the primary one.
Nevertheless, in a primary-less-noisy DM-CIC,
the primary receiver is assumed to have more information than the cognitive receiver,
as \eqref{eq:defn1} implies. This condition could particularly arise in the cognitive Z-interference channel
in which the link from the primary user to the cognitive receiver is absent.
For example, one can verify that the capacity result for the GCZIC at very
strong interference \cite[Corollary 4]{vaezi2011superposition}
is the counterpart of Corollary~\ref{cor2}, for Gaussian inputs.
This is also shown independently in \cite[Theorem V.2]{Rini3}.

%
%
%

\section{Appendix}
\label{anx}

\subsection{Proof of Theorem \ref{thm2}}
\label{anx1}

\begin{proof}
We prove this theorem by showing the code construction, encoding, decoding,
and error analysis.

\subsubsection{Code construction}
Fix $p(w,x_1)$ and $p(x_2|w,x_1)$. Randomly and
independently generate $2^{nR_{1}}$ sequences $(w^n(m_1),x_1^n(m_1))$, $m_1
\in [1: 2^{nR_{1}}]$ {\textit{i.i.d.}} according to $\prod_{i=1}^np_{WX_1}(w_{i},x_{1i})$.
 Next, for each sequence $(w^n(m_1), x_1^n(m_1))$,
randomly and conditionally independently generate $2^{nR_{2}}$ sequences
$x_2^n(m_1, m_2)$, $m_2 \in [1: 2^{nR_{2}}]$, with {\textit{i.i.d.}} elements according to
$\prod_{i=1}^np_{X_2|WX_1}(x_{2i}|w_{i}(m_1)x_{1i}(m_1))$.

\subsubsection{Encoding}
To send messages $(m_1,m_2)$, the primary transmitter
sends the codeword $x_1^n(m_1)$ whereas the secondary transmitter
sends the codeword $x_2^n(m_1,m_2)$.

\subsubsection{Decoding}
We use {\it joint typicality} for decoding. The cognitive
receiver ($Y_2$) can decode both messages whereas the other receiver can only decode one of them, namely $m_1$.
Decoder 1 declares that message $\hat{m}_{1}$ is sent if it is
the unique message such that $(w^n(\hat{m}_1),x_{1}^n(\hat{m}_1), y^n_1) \in {\cal
T}^{(n)}_\epsilon$. Likewise, decoder 2
declares that message $\hat{\hat{m}}_{2}$ is sent if it is the unique message
such that $(w^n({m}_1), x_1^n({m}_1),
x_2^n({m}_1, \hat{\hat{m}}_2), y^n_2) \in {\cal T}^{(n)}_\epsilon$, for some ${m}_1$. In other cases, as analyzed below, the decoders declare error.

\subsubsection{Error Analysis}
Without loss of generality, we assume that
$(M_1,M_2)=(1,1)$ is sent in order to analyze the probability of error. To evaluate the average probability of
error for decoder 1, we define the following error events
\begin{align*}
 E_{11} &= (W^n(1),X_1^n(1), Y^n_1) \notin {\cal T}^{(n)}_\epsilon, \nonumber \\
 E_{12} &= (W^n(m_1),X_1^n(m_1), Y^n_1) \in {\cal T}^{(n)}_\epsilon  \text { for } m_1 \neq 1.
  \label{eq:E1}
\end{align*}
Then, by using union bound, the probability of error for decoder 1 is upper bounded by
\begin{align}
P(E_{1}) = P(E_{11} \cup E_{12} ) \leq P(E_{11}) + P(E_{12} ).
\end{align}
But, $ P(E_{11}) \rightarrow
0 \text { as }  n \rightarrow \infty$, by the law of large numbers (LLN).  Moreover, since for $ m_1 \neq 1 $, $(W^n(m_1),X_1^n(m_1))$
is independent of $(W^n(1),X_1^n(1), Y^n_1)$, by the {\it packing
lemma} \cite{ElGamal},
 $ P(E_{12}) \rightarrow 0 \text { as }  n \rightarrow \infty \text { if }  R_1 \leq I(W,X_1;Y_1)-\delta(\epsilon). $

To evaluate the average probability of error for decoder 2, we define the following error events
\begin{align*}
 E_{21} =& (W^n(1), X_1^n(1), X_2^n(1,1), Y^n_2) \notin {\cal T}^{(n)}_\epsilon, \nonumber \\
 E_{22} =& (W^n(1), X_1^n(1), X_2^n(1,m_2), Y^n_2) \in {\cal
 T}^{(n)}_\epsilon \nonumber \\
 &\text { for some }  m_2 \neq 1,  \nonumber \\
 E_{23} =& (W^n( m_1), X_1^n( m_1),X_2^n(m_1,m_2), Y^n_2) \in {\cal T}^{(n)}_\epsilon  \nonumber \\
 &\text { for some } m_1 \neq 1, m_2 \neq 1.
  \label{eq:E2}
\end{align*}
Using union bound, the probability of error for decoder 1 is bounded as
\begin{align}
 P(E_{2}) &= P(E_{21} \cup E_{22} \cup E_{23} ) \nonumber \\
 &\leq P(E_{21}) + P(E_{22}) +P(E_{23} ).
\end{align}
Now, we evaluate the terms in the right-hand side (RHS) of this inequality when $ n \rightarrow
\infty$. First, by the LLN
$ P(E_{21}) \rightarrow 0 \text { as }  n \rightarrow
\infty. $ Then, for $ m_2 \neq 1 $,
$X_2^n(1,m_2)$ is conditionally independent of $Y^n_2$
given $(W^n(1),X^n_1(1))$. Thus, by the packing lemma
 $ P(E_{22}) \rightarrow 0 \text { as }  n \rightarrow
 \infty \text { given }  R_2 \leq I( X_2;Y_2|W,X_1)-\delta(\epsilon)$.
Finally consider $E_{23}$; for $ m_1 \neq 1 $ and $ m_2 \neq 1 $,
 $(W^n(m_1), X^n_1(m_1), X^n_2(m_1,m_2))$ is independent of
$Y^n_2$. Again, by the packing lemma
 $ P(E_{23}) \rightarrow 0 \text { as }
 n \rightarrow \infty \text { if  }  R_1 + R_2 \leq
 I(W,X_1,X_2;Y_2)-\delta(\epsilon) = I(X_1,X_2;Y_2)
  -\delta(\epsilon)$; the equality follows since $W \rightarrow X_1,X_2 \rightarrow Y_2$
  forms a Markov chain.
The proof of achievability is completed by the above analysis. That is, if \eqref{eq:inner2} is satisfied, both receivers
can decode corresponding messages with the total probability of error tending to
zero. Therefore,
there exists a sequence of good codes for which error probability goes
to 0.
\end{proof}

\subsection{Proof of Theorem \ref{thm3}} 
\label{anx2}

\begin{proof}
We prove this theorem by showing the code construction, encoding, decoding,
and error analysis.

\subsubsection{Code construction}
Fix $p(x_1)$ and $p(w,x_2)$. Randomly and
independently generate $2^{nR_{1}}$ sequences $x_1^n(m_1)$, $m_1
\in [1: 2^{nR_{1}}]$ {\textit{i.i.d.}} according to $\prod_{i=1}^np_{X_1}x_{1i}$.
 Also, for each $x_1$,
randomly and independently generate $2^{nR_{2}}$ sequences
$w^n(m_1,m_2)x_2^n(m_1,m_2)$, $m_2 \in [1: 2^{nR_{2}}]$, with {\textit{i.i.d.}} elements according to
$\prod_{i=1}^np_{WX_2}w_{i}(m_1,m_2)x_{2i}(m_1,m_2)$.

\subsubsection{Encoding}
To send messages $(m_1,m_2)$, the primary and cognitive transmitters, respectively,
send the codewords $x_1^n(m_1)$ and  $x_2^n(m_1,m_2)$.

\subsubsection{Decoding}
We use {\it joint typicality} for decoding, where the primary
receiver can decode both messages whereas the cognitive receiver can only decode $m_2$.
Decoder 2 declares that message $\hat{m}_{2}$ is sent if it is
the unique message such that $(w^n({m}_1,\hat{m}_2),x_{2}^n({m}_1,\hat{m}_2), y^n_2) \in {\cal
T}^{(n)}_\epsilon$, for some ${m}_1$. Similarly, decoder 1
declares that message $\hat{\hat{m}}_{1}$ is sent if it is the unique message
such that $(w^n(\hat{\hat{m}}_1, {m}_2), x_2^n(\hat{\hat{m}}_1, {m}_2),
x_1^n(\hat{\hat{m}}_1), y^n_2) \in {\cal T}^{(n)}_\epsilon$. In other cases, the decoders declare error.

\subsubsection{Error Analysis}
Error analysis is very similar to that of Theorem~\ref{thm2} and is omitted here.
\end{proof}

\subsection{Proof of Lemma \ref{lem1}}
\label{anx3}

\begin{proof}
The Lemma is similar to \cite[Lemma 5]{Maric2}. We can write
\begin{align}
I(U;Y_1|X_1) &= \sum_{x_1}p(x_1)I(U;Y_1|X_1=x_1) \nonumber \\
& \leq \sum_{x_1}p(x_1)I(U;Y_2|X_1=x_1) \nonumber \\
 &= I(U;Y_2|X_1)
\end{align}
the inequality follows because $ I(U;Y_1) \leq I(U;Y_2)$ holds for all joint distributions $p(u,x_1,x_2)$.
\end{proof}



\section*{Acknowledgement}
The author would like to thank Fabrice Labeau for his invaluable support and comments.

\end{document}